\pdfoutput=1

\documentclass[conference,letterpaper]{IEEEtran}

\usepackage[utf8]{inputenc} 
\usepackage[T1]{fontenc}
\usepackage{url}
\usepackage{ifthen}
\usepackage{xcolor}
\usepackage{cite}
\usepackage{booktabs}
\usepackage{makecell}
\usepackage[cmex10]{amsmath} 
\usepackage{amsthm}
\usepackage{nicefrac}
\usepackage[normalem]{ulem}

\theoremstyle{definition}
\newtheorem{definition}{Definition}[section]

\usepackage{cleveref}
\crefname{algocf}{alg.}{algs.}
\Crefname{algocf}{Algorithm}{Algorithms}

\usepackage[ruled, vlined]{algorithm2e}

\usepackage{amsmath,amsfonts,bm}

\def\eqref#1{equation~\ref{#1}}

\def\1{\bm{1}}

\def\rva{{\mathbf{a}}}
\def\rvb{{\mathbf{b}}}

\def\rvh{{\mathbf{h}}}

\def\rvx{{\mathbf{x}}}
\def\rvy{{\mathbf{y}}}
\def\rvz{{\mathbf{z}}}

\def\rmA{{\mathbf{A}}}

\DeclareMathAlphabet{\mathsfit}{\encodingdefault}{\sfdefault}{m}{sl}
\SetMathAlphabet{\mathsfit}{bold}{\encodingdefault}{\sfdefault}{bx}{n}

\newcommand{\KL}{D_{\mathrm{KL}}}
\newcommand{\Var}{\mathrm{Var}}

\DeclareMathOperator*{\argmax}{arg\,max}

\usepackage{amsthm}
\usepackage{mathtools}

\newtheorem{theorem}{Theorem}[section]
\newtheorem{lemma}[theorem]{Lemma}

\newtheorem{proposition}[theorem]{Proposition}
\theoremstyle{remark}
\newtheorem*{remark}{Remark}

\newcommand{\poly}{\mathrm{poly}}
\newcommand{\Ent}{\mathbb{H}}

\DeclarePairedDelimiterX{\infdivx}[2]{[}{]}{%
  #1\;\delimsize\|\;#2%
}
\DeclarePairedDelimiterX{\infdivmi}[2]{[}{]}{%
  #1\;\delimsize;\;#2%
}
\newcommand{\MI}{\mathbb{I}\infdivmi}

\newcommand{\KLD}{\KL\infdivx}
\newcommand{\TVD}{D_{TV}\infdivx}
\newcommand{\infD}{D_{\infty}\infdivx}

\newcommand{\fD}{D_{f}\infdivx}

\DeclarePairedDelimiter{\norm}{\lVert}{\rVert}
\DeclarePairedDelimiter{\abs}{\lvert}{\rvert}
\DeclarePairedDelimiterX{\innerProd}[2]{\langle}{\rangle}{%
    #1,#2%
}

\def\Oh{\mathcal{O}}

\renewcommand{\Var}{\mathbb{V}}
\newcommand{\Reals}{\mathbb{R}}

\newcommand{\Nats}{\mathbb{N}}

\newcommand{\defeq}{\stackrel{def}{=}}
\newcommand{\eqdef}{=\vcentcolon} 
\newcommand{\Prob}{\mathbb{P}}

\newcommand{\Ind}{\mathbf{1}}
\newcommand{\Unif}{\mathrm{Unif}}
\newcommand{\Exp}{\mathbb{E}}

\newcommand{\logtwo}{\log_2}
\newcommand{\CommRand}{\mathcal{S}}
\newcommand{\exptwo}{\exp_2}

\newcommand{\Law}[1]{\operatorname{Law}\left({#1}\right)}

\newcommand{\Algorithm}{\mathcal{A}}

\begin{document}
\title{Some Notes on the Sample Complexity of Approximate Channel Simulation} 

\author{%
  \IEEEauthorblockN{Gergely Flamich}
  \IEEEauthorblockA{University of Cambridge \\
  Cambridge, UK \\
  gf332@cam.ac.uk}
  \and
  \IEEEauthorblockN{Lennie Wells}
  \IEEEauthorblockA{University of Cambridge \\
  Cambridge, UK \\
  ww347@cam.ac.uk}
}

\maketitle

\begin{abstract}
Channel simulation algorithms can efficiently encode random samples from a prescribed target distribution $Q$ and find applications in machine learning-based lossy data compression. 
However, algorithms that encode exact samples usually have random runtime, limiting their applicability when a consistent encoding time is desirable. 
Thus, this paper considers approximate schemes with a fixed runtime instead. 
First, we strengthen a result of Agustsson and Theis \cite{agustsson2020universally} and show that there is a class of pairs of target distribution $Q$ and coding distribution $P$, for which the runtime of any approximate scheme scales at least super-polynomially in $\infD{Q}{P}$.
We then show, by contrast, that if we have access to an unnormalised Radon-Nikodym derivative $r \propto dQ/dP$ and knowledge of $\KLD{Q}{P}$, we can exploit global-bound, depth-limited A* coding \cite{flamich2022fast} to ensure $\TVD{Q}{P} \leq \epsilon$ and maintain optimal coding performance with a sample complexity of only $\exp_2\!\big((\,\KLD{Q}{P} + o(1)\,) \big/ \epsilon\big)$.
\end{abstract}
\section{Introduction}
\noindent
One-shot channel simulation is a communication problem between two parties, Alice and Bob, who share a probabilistic model over a pair of correlated random variables $\rvx, \rvy \sim P_{\rvx, \rvy}$, as well as a source of common randomness.
In one round of communication, Alice receives a sample $\rvy \sim P_\rvy$ and needs to send the minimum number of bits to Bob, so that he can simulate a sample $\rvx \sim P_{\rvx \mid \rvy}$.
\par
Solutions to this problem provide an alternative to quantization and entropy coding for implementing transform coding.
Thus, efficient channel simulation protocols have far-reaching applications in machine learning-based data compression, as we can use them to turn essentially any generative model into a lossy compression algorithm \cite{havasi2018minimal, flamich2020compressing, agustsson2020universally,guo2023compression,he2024recombiner}.
Furthermore, channel simulation provides unique advantages over quantization-based transform coding in many scenarios, such as when in addition to the rate-distortion trade-off we consider realism constraints \cite{Theis2021a, theis2022lossy} or differential privacy \cite{shah2021optimal}.
\par
Unfortunately, under a reasonable computational model of sampling, exact channel simulation in general is hopelessly difficult.
Concretely, for a given $\rvy \sim P_{\rvy}$, let us set $P \gets P_\rvx$ and $Q \gets P_{\rvx \mid \rvy}$ for brevity.
The standard computational model of channel simulation protocols assumes that Alice and Bob's shared randomness takes the form of an infinite sequence of i.i.d.\ $P$-distributed samples $(X_1, X_2, \hdots)$ and Alice has to select an index $N$, such that $X_N \sim Q$.
Then, under mild assumptions on the selection rule, Goc and Flamich \cite{goc2024channel} show that the sample complexity, i.e.\ the number of samples Alice needs to examine on average from the sequence before she can determine $N$, is at least $\exptwo(\infD{Q}{P})$, where $\infD{Q}{P}$ is the R{\'e}nyi $\infty$-divergence and $\exptwo(x) = 2^x$.
\par
It is thus natural to ask whether relaxing the requirement that the law $P_{X_N}$ of the selected sample be exactly $Q$ could help reduce the computational complexity of channel simulation algorithms.
This question is related to previous investigations by Chatterjee and Diaconis \cite{chatterjee2018sample}, Agustsson and Theis \cite{agustsson2020universally}, and Block and Polyanskiy \cite{block2023sample}, who considered approximate sampling without regard for how efficiently the sample can be encoded.
In this paper, we build on these works and strengthen some of their relevant results.
Interestingly, we find that approximate channel simulation is not harder than approximate sampling.
In fact, taking inspiration from the channel simulation literature we can improve sample complexity bounds in general, as we demonstrate in \Cref{sec:stronger_block_polyanskiy}.
\par
\textit{Contributions.}
The goal of our paper is to determine the sample complexity required for approximate channel simulation under different computational assumptions.
In particular,
\begin{enumerate}
\item We strengthen a result of Agustsson and Theis, and show that approximate sampling is prohibitively expensive for general distributions.
Concretely, we show that under the standard complexity-theoretic assumption that $P \neq RP$, there is no algorithm whose runtime scales polynomially in $\infD{Q}{P}$ and which can output an approximate sample with law $\widetilde{Q}$ such that $\TVD{\widetilde{Q}}{Q} \leq 1/12$.
\item
We give an improved variant of Block and Polyanskiy's approximate rejection sampler \cite{block2023sample}, that can achieve $\TVD{\widetilde{Q}}{Q} \leq \epsilon$ with a sample complexity of 
\begin{align*}
\ln\left(\frac{1}{(1-\gamma)\epsilon}\right) (f')^{-1}\left(\frac{\fD{Q}{P}}{\gamma \epsilon}\!\right)
\end{align*}
for any $f$-divergence and $\gamma \in (0, 1)$.
While this might seem to contradict our first result, we clarify that this requires \textit{exact} knowledge of $\fD{Q}{P}$ and $dQ/dP$.
\item 
We demonstrate that global-bound, depth-limited A* coding \cite{flamich2022fast} can achieve $\TVD{\widetilde{Q}}{Q} \leq \epsilon$ error with a sample complexity of $\exptwo((\KLD{Q}{P} + c)/\epsilon)$ for $c = e^{-1}\logtwo e + 1$. 
\end{enumerate}
\section{Background}
\noindent
\textit{Notation.}
For two real numbers $a, b$ we define the infix notation $a \wedge b = \min\{a, b\}$ and $a \vee b = \max\{a, b\}$.
We denote the base two logarithm as $\logtwo$ and its inverse function as $\exptwo$.
Likewise, we denote the natural logarithm as $\ln$ and its inverse function as $\exp$.
Let $Q$ and $P$ be probability measures over the measurable space $(\Omega, \mathcal{A})$; we will always assume that $\Omega$ is Polish.
Then, we define their total variation distance as $\TVD{Q}{P} = \sup_{B \in \mathcal{A}}\abs{Q(B) - P(B)}$.
Furthermore, assuming $Q \ll P$ and denoting their Radon-Nikodym derivative as $r = dQ/dP$, we define the Kullback-Leibler divergence of $Q$ from $P$ as $\KLD{Q}{P} = \Exp_{X \sim Q}[\logtwo r(X)]$ and the R{\'e}nyi $\infty$-divergence as $\infD{Q}{P} = \logtwo \norm{r}_\infty$, where $\norm{\cdot}_{\infty}$ is the $P$-essential supremum of a $P$-measurable function.
Finally, let $\mathcal{F}\! =\! \{f\!:\! [0, \infty) \to \Reals_{\geq 0} \cup \{\!\infty\!\} \mid f \text{ convex}, f(1)\! =\! 0, f'(1)\! =\! 0\}$.
Then, for an $f \in \mathcal{F}$ we define the $f$-divergence of $Q$ from $P$ as $\fD{Q}{P} = \Exp_{X \sim P}[f(r(X))]$.
\par
\textit{One-shot channel simulation (OSCS)}, also known as relative entropy coding \cite{flamich2020compressing} or reverse channel coding \cite{theis2021algorithms}, is a communication problem between two parties, Alice and Bob, and is defined as follows.
Let $\rvx, \rvy \sim P_{\rvx, \rvy}$ be a pair of random variables, whose law is known to both parties.
Furthermore, we assume that Alice and Bob share a source of common randomness $\CommRand$.
In one round of channel simulation, Alice receives a symbol $\rvy \sim P_\rvy$ and sends the minimum number of bits to Bob such that he can simulate a sample $\rvx \sim P_{\rvx \mid \rvy}$ using $\CommRand$.
Surprisingly, it can be shown that Alice needs to send only $\MI{\rvx}{\rvy} + \logtwo(\MI{\rvx}{\rvy} + 1) + 4.732$ bits on average to achieve this \cite{li2021unified}.
\par
\textit{Selection samplers and sample complexity.}
In practice, we are also concerned with the encoding time of the channel simulation algorithm.
However, we first need to establish a reasonable model of computation within which we can make sense of runtime.
One could use notions from computational complexity theory, where runtime can be associated with the number of steps taken by a universal Turing machine that executes the sampling algorithm, we consider such a framework in \Cref{sec:stronger_lucas_result}.
However, we are also interested in purely statistical or information theoretical properties of sampling; we focus on a natural formulation of sample complexity for a certain class of sampling algorithms. 
We define this class of algorithms below, taking inspiration from \cite[Definition A.2]{goc2024channel}.
\begin{definition}[Selection samplers]
\label{def:a_star_like_exact_sampler}
Let $Q \ll P$ be probability measures over some space $\Omega$.
Let $(X_i)_{i\in\Nats}$ be a sequence of i.i.d.\ $P$-distributed random variables.
A selection sampler 
\textit{selects}an index, modeled by some random variable $N$ over $\Nats$, and returns a sample $X_N$.
Moreover,
\begin{itemize}
\item 
if $\Law{X_N} = Q$ we say that the sampler is \textbf{exact}.
\item if $\TVD{\Law{X_N}}{Q} \leq \epsilon$ for some $\epsilon > 0$ we say that the sampler is \textbf{$\epsilon$-approximate}.
\end{itemize}
%
Finally, if there is a constant $k \in \Nats$ such that $N \leq k$ then we 
call $k$ the sample complexity of the sampler.
If there is a stopping time $K$ adapted to the sequence $(X_i)_{i\in\Nats}$ such that $1 \leq N \leq K$ then we say that the sampler is A*-like and call $\Exp[K]$ its sample complexity.
\end{definition}
Intuitively, an A*-like sampler needs to examine $K$ proposal samples, after which it \textit{has to} select one of the samples it already examined to output a sample from the target.
We can now precisely state the result of Goc and Flamich \cite{goc2024channel} mentioned in the introduction: the sample complexity of any exact A*-like sampler is at least $\exp_2(\infD{Q}{P})$.
%
%
\par
\textit{A* coding.}
We now briefly describe global-bound A* coding \cite{flamich2022fast}, the namesake of \Cref{def:a_star_like_exact_sampler}, as we will utilise it in \Cref{sec:stronger_block_polyanskiy,sec:approximate_channel_simulation}.
The algorithm is equivalent to the Poisson functional representation \cite{li2018strong} and is based on A* sampling \cite{maddison2014sampling}.
Given the shared sequence $(X_i)_{i \in \Nats}$ of i.i.d.\ $P$-distributed samples and the (potentially unnormalized) Radon-Nikodym derivative $\widetilde{r} \propto dQ/dP$, A* coding selects the sample with index $N = \argmax_{k \in \Nats}\{\ln\widetilde{r}(X_k) + G_k \}$, where $G_1$ is a Gumbel random variable with mean $0$ and scale $1$, and for $k > 1$ each $G_k \mid G_{k - 1}$ is a standard Gumbel random variable truncated to $(-\infty, G_{k - 1})$.
While the maximisation to select $N$ is over all positive integers, it can be shown \cite{maddison2016poisson} that it is sufficient to examine the first $K$ elements of the sequence, where $K$ is a geometric random variable with mean $\norm{\widetilde{r}}_\infty$.
\section{Why characterise approximate channel simulation using total variation?}\label{sec:why-total-variation}
\par
A crucial detail to consider in studying approximate channel simulation protocols is how we ought to measure the approximation error.
In this section, we present some arguments for why total variation distance is an appropriate choice.
\par
Our main motivation stems from the following ``one-shot'' interpretation of the total variation distance.
Let a sample $X$ follow distribution $Q$ or $\widetilde{Q}$ with probability $1/2$.
Then the probability that an optimal observer can successfully tell which of the two distribution a given sample $X$ follows is
\cite{nielsen2013hypothesis,blau2018perception}:
\begin{align}
p_{\text{success}} = \frac{1}{2}\TVD{Q}{\widetilde{Q}} + \frac{1}{2}.
\nonumber
\end{align}
Identifying $Q$ with the target distribution of a channel simulation algorithm and $\widetilde{Q}$ with its output distribution, we see the above characterisation aligns well with our goal: we are interested in the quality of a single encoded sample, as opposed to the quality of quantities derived from samples.
\par
Furthermore, providing guarantees on the total variation integrates well with our main application of interest: lossy data compression with realism constraints, or the rate-distortion-perception trade-off \cite{blau2018perception,Theis2021a,theis2024makes}.
For this argument, we briefly describe how lossy data compression is usually implemented using transform coding.
We first encode a stochastic representation $\rvx\! \sim\! Q$ of some data $\rvy\! \sim\! P_\rvy$ using a channel simulation algorithm, and use some measurable transformation $g$ (usually a neural network) to recover the data: $\rvy'\! =\! g(\rvx)$.
Now, in addition to the rate and the distortion of our compressor, we are also interested in controlling its \textit{output distribution}, 
given by the pushforward measure $P_{\rvy'} = g_*Q$.
In the usual, adversarial formulation of realism we require
$\TVD{P_\rvy}{P_{\rvy'}} \leq \delta$, where $\delta = 0$ corresponds to perfect realism \cite{blau2018perception,Theis2021a}.

Now, assume that we wish to use an $\epsilon$-approximate scheme to encode a sample from $Q$ instead of an exact one, resulting in a sample with distribution $\widetilde{Q}$ with $\TVD{\widetilde{Q}}{Q} \leq \epsilon$.
Then the output distribution is given by $g_*\widetilde{Q}$, so the realism of our transform coder depends on both the transform $g$ and the channel simulation protocol we use.
By applying the triangle inequality, we can bound the realism error as
\begin{align}
\TVD{P_\rvy}{g_* \widetilde{Q}} 
&\leq \TVD{P_\rvy}{g_* Q} + \TVD{g_* \widetilde{Q}}{g_* Q} \nonumber\\
&\leq \TVD{P_\rvy}{g_*Q} + \TVD{\widetilde{Q}}{Q} \nonumber\\
&\leq \delta + \epsilon,
\nonumber
\end{align}
where the second inequality uses the data processing inequality and the third applies the assumed bounds on the TV distances.
The practical significance of this decomposition is that in machine-learning-based pipelines, where $g$ is usually a neural network and we wish to learn its parameters, it provides principled justification for optimizing $\TVD{P_\rvy}{g_*Q}$, which can be done relatively easily using adversarial methods \cite{mentzer2020high}, instead of $\TVD{P_\rvy}{g_* \widetilde{Q}}$.
\par
Finally, on a more technical note, we consider an error bound proposed by Chatterjee and Diaconis \cite{chatterjee2018sample} for self-normalized importance sampling that has been adapted to minimal random coding \cite{havasi2018minimal} and A* coding \cite{flamich2022fast}.
To state the bound, write $\widetilde{Q}_n$ for the approximate distribution of the procedure; consider a sample complexity of $n = \exptwo(\KLD{Q}{P} + t)$ for some $t \geq 0$; then for any measurable function $f$, we have
\begin{gather}
\Prob\left[\abs*{\Exp_{\widetilde{Z} \sim \widetilde{Q}_n}[f(\widetilde{Z})] - \Exp_{Z \sim Q}[f(Z)]} \geq \frac{2 \norm{f} \epsilon}{1 - \epsilon}\right] \leq 2 \epsilon, \label{eq:chatterjee_prob_bound} \\
\epsilon = \left(2^\frac{-t}{2} + 2 \sqrt{\Prob_{X \sim Q}\left[\logtwo \frac{dQ}{dP}(X) \! \geq \!\KLD{Q}{P}\! + \!\frac{t}{2}\right]}\right)^{\frac{1}{2}}
\nonumber
\end{gather}
and $\norm{f}$ denotes the $L^2(Q)$ norm of $f$.
It might appear that the error $\epsilon$ vanishes exponentially quickly in the quantity $t$, which was introduced as an additive overhead to $\KLD{Q}{P}$, at odds with the results from \cite{agustsson2020universally,block2023sample}.
To resolve this conflict, first note that this bound applies to a test function $f$ rather than providing a one-shot guarantee on the approximate sample.
Secondly, as noted in \cite{chatterjee2018sample}, the bound is only meaningful when $\logtwo \frac{dQ}{dP}(X)$ is concentrated; we present a simple calculation in Appendix~\ref{app:diaconis-looseness} that shows that for certain pairs of distributions, ensuring a given tolerance in fact requires $t$ to scale with $\KLD{Q}{P}$, corresponding to a much larger sample size.
In Appendix~\ref{app:diaconis-looseness} we also present a (strict) strengthening of this bound inspired by ideas that we will develop in \Cref{sec:stronger_block_polyanskiy}.
\section{Improving a result of Agustsson and Theis}
\label{sec:stronger_lucas_result}
\noindent
In this section, we strengthen the result of \cite{agustsson2020universally} on the computational hardness of approximate sampling.
However, before we begin, we make some definitions that we will use later.
\begin{definition}[Restricted Boltzmann Machine (RBM)]
For some integer $M$, let $\Omega = \{0, 1\}^M$ and let $\rva, \rvb \in \Reals^M$ and $\rmA \in \Reals^{M\times M}$.
Then, the RBM distribution $Q$ with parameters $\theta = \{\rva, \rvb, \rmA\}$ is given by the probability mass function
\begin{align}\label{eq:RBM-density}
q(\rvz) \propto \sum_{\rvh \in \Omega} \exp\left(\rva^\top \rvz + \rvh^\top \rmA \rvz + \rvb^\top \rvh \right).
\end{align}
\end{definition}
\begin{definition}[Efficiently evaluatable representation]
\label{def:efficiently_evaluatable_rep}
For positive integers $N, M$ a function $f: \{0, 1\}^N \to \{0, 1\}^M$ is an efficiently evaluatable representation of a distribution $Q$ if:
\begin{enumerate}
\item $f$ consits of a Boolean circuit of $\poly(M)$ size and ${N = \poly(M)}$ input bits and $M$ output bits.
\item For $B \sim \Unif\left(\{0, 1\}^N\right)$, we have $f(B) \sim Q$.
\end{enumerate}
\end{definition}
%
\begin{theorem}
\label{thm:improved_sampling_hardness}
Consider an algorithm which receives the parameters of an arbitrary RBM $Q$ of problem size $M$ as input and has access to an unlimited number of i.i.d.\ random variables $Z_n\! \sim\! P$, where $P$ is the uniform measure over $\{0,1\}^M$. 
It outputs $\widetilde{Z}\! \sim\! \widetilde{Q}$ with $\TVD*{\widetilde{Q}}{Q} \!\leq\! 1/12$.
If $RP \neq NP$, there is no such algorithm with $\poly(\infD{Q}{P})$ time complexity.
\end{theorem}
\begin{proof}
The proof follows the proof of Agustsson and Theis \cite{agustsson2020universally} mutatis mutandis, which we repeat here for completeness.
\par
The high-level idea is that for such pairs of distributions $Q \ll P$, even evaluating their Radon-Nikodym derivative $r = dQ/dP$ is difficult.
To this end, we make use of the result of Theorem 13 of Long and Servedio \cite{long2010restricted}:
\begin{theorem}
\label{thm:rbm_simulation_hardness}
If $RP \neq NP$, then there is no polynomial-time algorithm with the following property: Given parameters $\theta = (\rmA, \rva, \rvb)$ as input, the algorithm outputs an efficiently evaluatable
representation of a distribution whose total variation distance from an RBM with parameters $\theta$ is at most $1 / 12$.
\end{theorem}
Now, fix some positive integer $M$ and fix some $\theta$ as given in \Cref{thm:rbm_simulation_hardness}.
Assume there is an algorithm $\Algorithm$ which outputs $\widetilde{Z} \sim \widetilde{Q}$ in $\psi(\infD{Q}{P})$ steps with $\TVD*{\widetilde{Q}}{Q} \leq 1/12$ for some polynomial $\psi$.
We compute
\begin{align}
\infD{Q}{P}
= \log_2 \max_{\rvz \in \Omega}\!\left\{\!\frac{q(\rvz)}{2^{-M}}\!\right\} 
\leq \log_2 \!\left\{\!\frac{1}{2^{-M}}\!\right\} = M \, .
\label{eq:rbm_inf_divergence_bound}
\end{align}
Then by \Cref{eq:rbm_inf_divergence_bound} the computational complexity of $\Algorithm$ is at most $N = \psi(M)$.
In that time, $\Algorithm$ can examine at most $N$ random variables $Z_n$; since the input random variables are i.i.d.,
we can assume without loss of generality that $\Algorithm$ examines the first $N$ of them. 
Since the proposal $P$ is the uniform measure on $\{0, 1\}^M$, these $N$ variates correspond to an input of $M \cdot \psi(M) = \poly(M)$ uniformly random bits.
This is an efficiently evaluatable representation of $\widetilde{Q}$, which contradicts \Cref{thm:rbm_simulation_hardness}, assuming $RP \neq NP$.
\end{proof}


%
\section{Improving the scheme of Block and Polyanskiy}
\label{sec:stronger_block_polyanskiy}
\begin{algorithm}[t]
\SetAlgoLined
\DontPrintSemicolon
\SetKwInOut{Input}{Input}
\SetKwInOut{Output}{Output}
\SetKwFunction{clip}{clip}
\Input{Sequence $(X_i)_{i \in \Nats}$ of i.i.d.\ $P$-distributed samples,
target $Q$ defined via $\widetilde{r} \propto dQ/dP$, 
computational budget $k$}
$N, Y, G_0, L \gets (0, \perp, \infty, -\infty)$\;
\For{$i = 1$ \KwTo $k$}{
$G_{i} \sim \text{TruncGumbel}(0, 1)\vert_{(-\infty, G_{i - 1})}$ \;
\If{$L < \ln\widetilde{r}(X_i) + G_i$}{
 $L \gets \ln\widetilde{r}(X_i) + G_i$ \;
 $N, Y \gets i, X_i$
 }
}
\KwRet{$N, Y$}
\caption{Depth-limited A* coding.}
\label{alg:modified_a_star_coding}
\end{algorithm}
\noindent
Recently, Block and Polyanskiy \cite{block2023sample} have shown that by modifying rejection sampling \cite{neumann1951various}, we can achieve $\epsilon$ error in total variation at a sample complexity of 
\begin{align}
k = \frac{2}{1 - \epsilon}\ln\left(\frac{2}{\epsilon}\right)(f')^{-1}\left(\frac{4 \cdot \fD{Q}{P}}{\epsilon}\right) \vee 2.
\label{eq:block_polyansky_sample_complexity}
\end{align}
In particular, as a special case of the above equation we get that $k = \Oh(\exptwo(4\cdot \KLD{Q}{P} / \epsilon))$.
Note that, in general, we can have $\KLD{Q}{P} \ll \infD{Q}{P}$ so this might seem at odds with \Cref{thm:improved_sampling_hardness}.
To reconcile these two results, we highlight that additional assumptions were needed in Block and Polyanskiy's scheme to achieve this improved sample complexity: we need to be able to compute $\fD{Q}{P}$ as well as evaluate $dQ/dP$ \textit{exactly}.
By contrast these are not given in the setup of \cite{agustsson2020universally}, and it is shown in \cite[Theorem 8]{long2010restricted} that it is computationally hard to even approximate the normalizing constants corresponding to densities in \Cref{eq:RBM-density} to within an exponentially large factor!
%
\par
\textit{The approximate rejection sampler:}
We now describe the three key ideas of Block and Polyanskiy to modify rejection sampling \cite{neumann1951various} to get an approximate scheme for some target $Q$ and proposal $P$:
\begin{enumerate}
\item They fix a budget $k \in \Nats$ and if the rejection sampler does not terminate, they pick one of the proposed samples at random.
Denoting the output distribution of the budgeted sampler as $\widetilde{Q}$ and its termination step as $K$, the sampler's error is
\begin{align}
\TVD{\widetilde{Q}}{Q} =  \Prob[K\! >\! k] \TVD{Q}{P} \leq \Prob[K\! >\! k]
\label{eq:approximate_vs_exact_tv_distance_bound}
\end{align}
\item Since $\Prob[K > k]$ in \Cref{eq:approximate_vs_exact_tv_distance_bound} depends on $\infD{Q}{P}$, they propose to use a truncated target $Q_M$, defined via its Radon-Nikodym derivative
\begin{align}
\frac{dQ_M}{dP}(x) \propto \Ind\left[\frac{dQ}{dP}(x) \leq M \right] \cdot \frac{dQ}{dP}(x).
\label{eq:block_polyanskiy_approx_target}
\end{align}
Then, using $Q_M$ in the rejection sampler with a budget of $k$ samples will have distribution $\widetilde{Q}_M$.
\item
They show that for a fixed $\epsilon > 0$, setting
\begin{align}
M = (f')^{-1}(4 \cdot \fD{Q}{P} / \epsilon)   
\nonumber 
\end{align}
and $k$ as in \Cref{eq:block_polyansky_sample_complexity} yields $\TVD{\widetilde{Q}_M}{Q_M} \leq \epsilon / 2$ and $\TVD{Q_M}{Q} \leq \epsilon / 2$.
Combining these two inequalities and applying the triangle inequality then yields the desired guarantee $\TVD*{\widetilde{Q}_M}{Q} \leq \epsilon$ at the sample complexity given in \Cref{eq:block_polyansky_sample_complexity}.
\end{enumerate}
\par
We now propose several small improvements to this scheme.
In this section, we focus on the sample complexity of the improved scheme and deal with encoding the approximate samples in \Cref{sec:approximate_channel_simulation}.
\par
\textit{Useful quantities and identities:} inspired by \cite{goc2024channel}, take ${w_Q(h) = \Prob_{X \sim Q}[r(X) \geq h]}$, ${w_P(h) = \Prob_{X \sim P}[r(X) \geq h]}$, and $W_P(h) = \int_0^h w_P(\eta) \, d\eta$, $S_P(h) = 1 - W_P(h)$.
Note, that by Fubini,
\vspace{-0.3cm}
\begin{align}
\int_0^\infty \overbrace{\Prob_{X \sim P}[r(X) \geq h]}^{= w_P(h)} \, dh = \Exp_{X \sim P}[r(X)] = 1.
\nonumber
\end{align}
This, taken together with the fact that $w_P \geq 0$ shows that we can interpret it as the probability density of a random variable $H$.
Thus, we can also interpret $W_P$ as $W_P(h) =  \Prob[H \leq h]$, and similarly $S_P(h) = \Prob[H>h]$ as $H$'s survival function.
Now, for any $f \in \mathcal{F}$ and $a \geq 1$, we can bound $S_P$ by noting
\begin{align*}
f'(a) \Prob[H > a] 
&= f'(a) \int_{a}^\infty w_P(h) \, dh \\
&\leq \int_{a}^\infty f'(h) w_P(h) \, dh \\
&\leq \int_{1}^\infty f'(h) w_P(h) \, dh \\
&= \int_\Omega \int_1^\infty f'(h) \Ind[r(x) \geq h] \,dh \,dP(x) \\
&= \int_\Omega \Ind[r(x) \geq 1] f(r(x)) \, dP(x) \\
&\leq \fD{Q}{P},
\end{align*}
where the second equality follows from Fubini (positive integrand), the third equality from the fundamental theorem of calculus, and we exploit $f\geq 0 = f(1)$.
Rearranging, we get 
\begin{align}
S_P(a) \leq \fD{Q}{P} / f'(a).
\label{eq:survival_prob_upper_bound}
\end{align}
Next, note that
\begin{align}
W_P(h)
&= \int_{\Omega} \int_{0}^h \Ind[r(x) \geq \eta] \, d\eta \, dP(x) \nonumber\\
&= \int_{\Omega} r(x) \wedge h \, dP(x) \label{eq:handy-wedge-result} \\
&= \int_{x: r(x) \geq h} h\, dP(x) + \int_{x: r(x) \leq h} r(x)\, dP(x) \nonumber\\
&= h \cdot w_P(h) + (1-w_Q(h)). \nonumber
\end{align}
Rearranging the terms, we find
\begin{align}
S_P(h) = w_Q(h) - h \cdot w_P(h).
\label{eq:survival_prob_identity}
\end{align}
\par
\textit{Better approximate target distribution for a tighter bound:}
The truncated target in \Cref{eq:block_polyanskiy_approx_target} is a rough approximation of $Q$.
Instead, we propose a better approximation that will also lend itself to simpler analysis and will allow us to improve \Cref{eq:block_polyansky_sample_complexity}.
We define this through its density
\begin{align}
\label{eq:improved-truncated-approx-target-density}
\frac{dQ_M}{dP}(x) = r_M(x) = \frac{r(x) \wedge M}{W_P(M)}.
\end{align}
That $W_P(M)$ normalises $r(x) \wedge M$ follows from \Cref{eq:handy-wedge-result}.
Now, set $\widetilde{M} =  M / W_P(M)$.
In Appendix~\ref{app:approximate-target-properties-and-optimality} we show that
\begin{align}
\TVD{Q_M}{Q} = S_P(\widetilde{M}) \stackrel{\text{\cref{eq:survival_prob_upper_bound}}}{\leq} \frac{\fD{Q}{P}}{f'(\widetilde{M})}
\label{eq:approximate_target_tv_bound}
\end{align}    
We now substitute this improved truncation into Block and Polyanskiy's scheme to improve upon \Cref{eq:block_polyansky_sample_complexity}.
Let $\epsilon$ be given, and let $\gamma \in (0,1)$ a constant that we have yet to choose.
By a standard result for rejection sampling, $K$ is geometrically distributed with mean $\norm{r_M}_\infty = M / W_P(M) = \widetilde{M}$.
Thus,
\begin{align}
\Prob[K > k] = \left(1 - \frac{ 1}{\widetilde{M}}\right)^k \leq \exp\left(-\frac{k}{\widetilde{M}}\right).
\end{align}
Setting $k \geq \widetilde{M} \ln(1 / (1-\gamma) \epsilon)$ thus gives $\TVD{\widetilde{Q}_M}{Q_M} \leq (1- \gamma) \epsilon$.
Furthermore, by \Cref{eq:approximate_target_tv_bound}, setting 
\begin{align}
\widetilde{M} = (f')^{-1}\left(\frac{\fD{Q}{P}}{\gamma \epsilon}\right)
\nonumber
\end{align}
will guarantee $\TVD{Q_M}{Q} \leq \gamma \epsilon$.
Combining these facts with the triangle inequality shows that for a fixed $\epsilon > 0$, a sample complexity of
\begin{align}
k \geq \ln\left(\frac{1}{(1-\gamma)\epsilon}\right)(f')^{-1}\left(\frac{\fD{Q}{P}}{\gamma \epsilon}\right)
\label{eq:better_general_sample_complexity}
\end{align}
is sufficient to achieve $\epsilon$ error in the TV distance.
Since the dependence on $(1-\gamma)$ is only logarithmic, the bound will often be tightest by taking $\gamma$ very close to 1.
This improves on \Cref{eq:block_polyansky_sample_complexity} by reducing the coefficient of $\fD{Q}{P}$ from 4 to a constant $1/\gamma$ that can be taken arbitrarily close to 1, removing the $2 / (1 - \epsilon)$ coefficient in the bound and additionally removing the requirement that $k$ be at least $2$.
In practice, we can compute $M$ from $\widetilde{M}$ by numerically inverting the function $h \mapsto h / W_P(h)$.
This function is strictly increasing on a suitable domain by a simple calculation in Appendix~\ref{app:approximate-target-properties-and-optimality}. 
In the rest of Appendix~\ref{app:approximate-target-properties-and-optimality}, we also derive cleaner bounds for exact rejection sampling from $Q_M$ (allowing random runtime), and show that our truncation is optimal in a certain sense. 

\par
\textit{Replacing rejection sampling by global-bound A* sampling yields identical analysis.}
Indeed, if we use the truncated target $Q_M$, the index $K$ at which we can guarantee termination is again geometrically distributed with mean $\norm{r_M}_\infty$ \cite{maddison2016poisson}.
Remarkably, we in fact only need access to an unnormalised version of $\tilde{r}_M \propto r_M$ to run the algorithm, though we would also need knowledge of $\fD{Q}{P}$ to compute the sample complexity in \Cref{eq:better_general_sample_complexity}.

\textit{Depth-limited A* sampling yields better bounds without needing truncation.}
Using depth-limited A* sampling, as described in \Cref{alg:modified_a_star_coding}, again only requires access to an unnormalised density ratio, but now is well behaved even when this density ratio is unbounded. Indeed, there is now no need for truncation; we next obtain tighter bounds by working with the original target distribution $Q$ and basing the bounds on the index $N$ of the accepted sample instead of the sampler's runtime $K$.
\par
As before, we use \Cref{eq:approximate_vs_exact_tv_distance_bound} to obtain the bound ${\TVD{\widetilde{Q}}{Q} \leq \Prob[N > k]}$ for a fixed sample complexity $k$.
Now, note that by Markov's inequality, we have
\begin{align}
\Prob[N > k] &= \Prob[\logtwo N > \logtwo k] \nonumber\\
&\leq \frac{\Exp[\logtwo N]}{ \logtwo k} \nonumber\\
&\leq \frac{\KLD{Q}{P} + e^{-1} \logtwo e + 1}{\logtwo k}, 
\label{eq:kl_bound_on_a_star_coding}
\end{align}
where the last equality follows from the identity given in Appendix A of \cite{li2018strong}.
Note that the bound in \Cref{eq:kl_bound_on_a_star_coding} does not depend on $\norm{r}_\infty$.
Therefore, for a fixed $\epsilon > 0$, choosing 
\begin{align}
k = \exptwo\left(\frac{\KLD{Q}{P} + e^{-1} \logtwo e + 1}{\epsilon}\right)
\label{eq:index_based_sample_complexity}
\end{align}
yields the desired bound $\TVD{\widetilde{Q}}{Q} \leq \epsilon$.
\Cref{eq:index_based_sample_complexity} significantly improves \Cref{eq:block_polyansky_sample_complexity}, and is close to being worst-case optimal \cite[Theorem 5]{block2023sample}.
It also significantly improves a result of Theis and Yosri \cite[Corollary 3.2]{theis2021algorithms}, which requires that $k = \Oh(\exptwo(\infD{Q}{P}))$ to guarantee the same $\epsilon$ total variation error.
\par
\textit{Remark.}
We can obtain almost-optimal $\epsilon$-approximate sample complexities similar to \Cref{eq:index_based_sample_complexity} by using ordered random coding \cite{theis2021algorithms}, greedy Poisson rejection sampling \cite{flamich2023gprs} or greedy rejection coding \cite{flamich2023adaptive,flamich2023grc} and modifying the arguments that yield \Cref{eq:kl_bound_on_a_star_coding} appropriately.
Furthermore, via similar arguments we can also obtain linear-in-the-KL sample complexities for the branch-and-bound variants of these algorithms when $\Omega$ is one-dimensional and $\tilde{r}$ is unimodal.
\section{Approximate Channel Simulation Using the Improved Scheme}
\label{sec:approximate_channel_simulation}
\noindent
Unfortunately, an issue with Block and Polyanskiy's sampler is that it cannot be used for channel simulation.
As we show in Appendix~\ref{app:lower_bound_on_rej_sampler_entropy},  $\Ent[N] \geq 4 \cdot \MI{\rvx}{\rvy} / \epsilon$, which is much worse than the optimal upper bound of $\MI{\rvx}{\rvy} + \logtwo(\MI{\rvx}{\rvy} + 1) + \Oh(1)$.
\par
Thankfully, this can also be fixed by replacing rejection sampling with A* coding.
For some fixed target $Q$ and proposal $P$, let $N$ denote the index returned by A* coding, and let $N'$ denote the index returned by \Cref{alg:modified_a_star_coding}, i.e.\ depth-limited A* coding with given a budget of $k$ samples.
Then, 
\begin{align}
N' &= \argmax_{n \in \{1,\hdots, k\}}\{\ln\widetilde{r}(X_k) + G_k\} \nonumber \\
&\leq \argmax_{k \in \Nats}\{\ln\widetilde{r}(X_k) + G_k \} = N \nonumber 
\end{align}
Hence, we have $\Exp[\logtwo N'] \leq \Exp[\logtwo N]$.
We encode $N'$ using the $\zeta$-distribution $\zeta(n \mid \lambda) \propto n^{-\lambda}$. 
Li and El Gamal \cite{li2018strong} show that by setting
${\lambda = 1 + 1 / (\MI{\rvx}{\rvy} + e^{-1}\logtwo e + 1)}$ we get
\begin{align}
\Ent[N'] < \MI{\rvx}{\rvy} + \logtwo(\MI{\rvx}{\rvy} + 1) + 4,
\end{align}
meaning we can encode our sample at the optimal rate.
%
\section{Discussion and Future Work}
\par
To summarise our results: $\epsilon$-approximate sampling has computational complexity super-polynomial in $\infD{Q}{P}$ without further assumptions.
When $\KLD{Q}{P}$ is known, one can get $\epsilon$-approximate samples with sample complexity exponential in $\KLD{Q}{P}/\epsilon$ and we can also encode these samples at the optimal rate using \Cref{alg:modified_a_star_coding}.
It would be interesting to consider what other structural assumptions may be leveraged to achieve lower sample complexity.


%
\noindent

\section*{Author contributions statement}
GF conceived the core of this work, was responsible for the majority of the content, and wrote the entire first draft.
LW helped GF craft the story of the paper, reworked the document, and contributed a number of insights, results and refinements.

\section*{Acknowledgments}
GF acknowledges funding from DeepMind.
LW is supported by the UK Engineering and Physical Sciences Research Council (EPSRC) under grant number EP/V52024X/1.

\bibliographystyle{ieeetr}
\bibliography{references}

\clearpage

\appendices

\vspace{20pt}
\section{Approximate target distribution: properties and optimality}
\label{app:approximate-target-properties-and-optimality}
\textit{Computing the TV distance of the approximate target.}
Recall that in \Cref{sec:stronger_block_polyanskiy}, we defined our truncated target distribution in \Cref{eq:improved-truncated-approx-target-density} via 
\begin{align*}
    \frac{dQ_M}{dP}(x) = r_M(x) = \frac{r(x) \wedge M}{W_P(M)} \, .
\end{align*}
We now compute $\TVD{Q_M}{Q}$.
By a standard characterisation of the total variation distance where both measures $Q_M, Q$ have densities with respect to a common dominating measure we have
\begin{align}\label{eq:truncated-TVD-characterisation-original}
    \TVD{Q_M}{Q}
    = \int_{x: r(x) > r_M(x)} (r(x) - r_M(x)) dP(x) 
\end{align}
To characterise where $r(x) > r_M(x)$, first note that $W_P(M) \leq 1$ so when $r(x) \leq M$ we have $r_M(x) = r(x)/W_P(M) \geq r(x)$.
When $r(x) \geq M$ we have $r_M(x) = M / W_P(M) \eqdef \tilde{M}$.
Therefore $r(x) > r_M(x)$ precisely when $r(x) > \tilde{M}$.
Plugging this into \Cref{eq:truncated-TVD-characterisation-original}, we get
\begin{align*}
    \TVD{Q_M}{Q} 
    &= \int_{x: r(x) > \tilde{M}} (r(x) - r_M(x)) dP(x) \\
    &= \int_{x: r(x) > \tilde{M}} (r(x) - \tilde{M}) dP(x)  \\
    &= w_Q(\tilde{M}) - \tilde{M} w_P(\tilde{M})
    = S_P(\tilde{M}) \, ,
\end{align*}
as claimed, where the final equality follows from \Cref{eq:survival_prob_identity} applied to $\tilde{M}$.

\textit{Further properties.}
Firstly, note that $w_P(h) > 0$ for $h \in [0,\norm{r}_\infty)$, so in fact $W_P(h)$ is strictly increasing on $[0,\norm{r}_\infty]$.
Correspondingly, $S_P$ is strictly decreasing on this domain, and it we can define its inverse $S_P^{-1}: [0,1] \to [0,\norm{r}_\infty]$.
\par
Secondly, $\phi: h \mapsto h/W_P(h)$ is differentiable for $h>0$ with derivative
\begin{align*}
    \phi'(h) 
    &= \frac{1}{W_P(h)} - \frac{h w_P(h)}{W_P(h)} \\
    &=\frac{1}{W_P(h)^2} \left(W_P(h) - h w_P(h) \right) \\
    &=\frac{1}{W_P(h)^2} \Prob_{X\sim Q}[r(X) < h]
    \,\geq\, 0
\end{align*}
where the final equality follows by \Cref{eq:survival_prob_identity} and the definition of $w_Q$. Moreover, the inequality is strict whenever $h$ is strictly bigger than the $Q$-essential infinum of $r$.

\textit{Sample complexity in expectation.}
If we are content to have an algorithm with unbounded random runtime then we can obtain improved bounds for the rejection sampling scheme described in \Cref{sec:stronger_block_polyanskiy} using our truncation.
Indeed, given an $\epsilon > 0$, we can take an $M$ such that $\tilde{M} = S_P^{-1}(\epsilon)$ (by the invertibility results of the previous subsection).
For such an $M$, the rejection sampler will be geometrically distributed with mean $\tilde{M} = S_P^{-1}(\epsilon)$. So \Cref{eq:survival_prob_upper_bound} immediately gives
\begin{align*}
        \Exp(K) = (S_P)^{-1}(\epsilon) \leq (f')^{-1}\left(\frac{\fD{Q}{P}}{\epsilon}\right) \,.
\end{align*}
In particular, for the case of the KL-divergence we obtain a precise upper bound of $\exptwo(\KLD{Q}{P}/\epsilon)$ with no additional $\gamma$ terms.
In the next subsection we show that this bound is optimal in a certain sense.

\textit{Optimality of the truncation for random run-times.}
The following lemma shows that the family of distributions defined by \Cref{eq:improved-truncated-approx-target-density} for $M > 0$ are Pareto-optimal for the dual objectives of minimising total variation distance from $Q$ and minimising R{\'e}nyi $\infty$-divergence from $P$.

\begin{lemma}\label{lem:pareto-optimal-approximation}
    Fix some $\epsilon>0$.
    Let $Q'$ be some other probability distribution on $\Omega$ such that $\TVD{Q'}{Q} \leq\epsilon$. Then $\exptwo\infD{Q'}{P} \geq S_P^{-1}(\epsilon)$.
\end{lemma}
\begin{proof}
Write $\tilde{M}' \defeq \exptwo\infD{Q'}{P}$ and $\tilde{M} = S_P^{-1}(\epsilon)$.
    Suppose for contradiction that $\tilde{M}' < \tilde{M}$.
    Then 
    \begin{align*}
        \TVD{Q'}{Q} 
        &= \int_{\{x: r(x) > r'(x)\}} r(x) - r'(x) dP(x)  \\
        &\geq \int_{\{x: r(x) \geq \tilde{M}'\}} r(x) - r'(x) dP(x)  \\
        &\geq \int_{\{x: r(x) \geq \tilde{M}'\}} r(x) - \tilde{M} dP(x) \\
        &= S_P(\tilde{M}')
        > S_P(\tilde{M}) = \epsilon\, ,
    \end{align*}
    where the final inequality follows because $\tilde{M}=S_P^{-1}(\epsilon) \leq \norm{r}_\infty$ and $S_P$ is strictly decreasing on $[0,\norm{r}_\infty]$.
    This yields the required contradiction.
\end{proof}

\begin{proposition}
Fix some $\epsilon > 0$.
Let $N,K$ correspond to an approximate A*-like non-causal sampler with sampling distribution $\hat{Q} = \Law{X_N}$ such that $\TVD{\hat{Q}}{Q} \leq \epsilon$.
Then 
\begin{align*}
    \Exp(K) \geq (S_P)^{-1}(\epsilon) \, .
\end{align*}    
\end{proposition}
\begin{proof}
    By \cite[Theorem A.2.]{goc2024channel} we have that $\Exp(K) \geq \exp(\infD{\hat{Q}}{P})$.
    Combining this with \Cref{lem:pareto-optimal-approximation} yields the result.
\end{proof}
\section{Discussion of importance sampling bounds of Chatterjee and Diaconis}\label{app:diaconis-looseness}
\textit{Looseness for certain pairs of distributions.}
The right hand side of the bound from \cite{chatterjee2018sample} we gave above crucially involves the quantity
\begin{align*}
    \Prob_{X \sim Q}[\logtwo r(X) > L + t/2] \, ,
\end{align*}
where, for the purposes of this appendix, we synthesise notation from the main-text with that of \cite{chatterjee2018sample} and write $L=\KLD{Q}{P}$.
\par
When $\logtwo r(Y)$ is supported on the set $\{0, L + t/2\}$, a calculation inspired by the equality case of Markov's inequality gives
\begin{align*}
(L + & t/2) \, \Exp_{Y \sim Q}\left[\Ind[\logtwo r(Y) \geq L + t/2]\right] \\
&= \Exp_{Y \sim Q}\left[\logtwo r(Y)\Ind[\logtwo r(Y) \geq L + t/2]\right] \\
&= \Exp_{Y \sim Q}\left[\logtwo r(Y)\Ind[\logtwo r(Y) \geq 0]\right] \\
&= \KLD{Q}{P} = L
\end{align*}
In particular, we see that 
\begin{align*}
\Prob_{X \sim Q}[\logtwo r(X) \geq L + t/2] = \frac{L}{L+t/2}
\end{align*}
So this term decays far slower than exponential in $t$.
Indeed, if we wish the bound in \Cref{eq:chatterjee_prob_bound} to be less than a given tolerance $2\epsilon$,
we must have
\begin{align*}
\epsilon^2 
\geq 2 \sqrt{\Prob_{X \sim Q}[\logtwo r(X) \geq L + t/2]}
\geq 2\sqrt{\frac{L}{L+t/2}}
\end{align*}
which upon rearranging gives
\begin{align*}
    t \geq 2L ((2/\epsilon^2)^2 - 1) \:.
\end{align*}
This gets large very quickly as $\epsilon$ gets small, and shows that the surplus of samples in the exponent in fact has to scale with $L=\KLD{P}{Q}$ in the accuracy desired.
\par
\textit{Tighter bound by construction analogous to \Cref{eq:improved-truncated-approx-target-density}.}
We now strictly strengthen the upper bounds from \cite{chatterjee2018sample} by using a tighter construction analogous to that of \Cref{eq:improved-truncated-approx-target-density}.
To state the bound, we need to introduce further notation from their original paper.
Indeed, for $\varphi \in L^2(Q)$ write
\begin{align*}
    I(\varphi) \defeq \int_\mathcal{X} \varphi(x) dQ(x), \quad I_n(\varphi) = \frac{1}{n}\sum_{i=1}^n \varphi(X_i) r(X_i) .
\end{align*}
Then their main result can be improved by the following theorem.
\begin{theorem}\label{thm:chatterjee-diaconis-strengthening}
Write $L = \KLD{Q}{P}$.
If $n=2^{L+t}$ for some $t\geq 0$, then
\begin{align*}
    \Exp(\abs{I_n(\varphi) - I(\varphi)}) \leq \norm{\varphi}_{L^2(Q)} \, \left(2^{-t/4} + 2 \sqrt{S_P(2^{L+t/2})}\right)
\end{align*} 
\end{theorem}
\begin{remark}
This is identical to the upper bound Theorem 1.1 of \cite{chatterjee2018sample} apart from the $S_P(\exptwo(L+t/2))$ term, which replaces a term of $\Prob_{X\sim Q}[\logtwo r(X) \geq L + t/2] = w_Q(\exptwo(L+t/2))$ in the original work. \Cref{eq:survival_prob_identity} shows that $S_P(h) \leq w_Q(h)$ for all $h\geq 0$ and so this is a strict strengthening of the original result.
\par
This immediately gives a corresponding strengthening of the upper bound of Theorem 1.2 of \cite{chatterjee2018sample} that we stated in \Cref{sec:why-total-variation}. We can replace the $w_Q(\exptwo(L+t/2))$ term by a $S_P(\exptwo(L+t/2))$ term by applying their union bound argument with \Cref{thm:chatterjee-diaconis-strengthening} in place of their upper bound from Theorem 1.1.
Moreover, by using \Cref{eq:survival_prob_upper_bound} we can get a further upper bound on this term from a bound on any $f$-divergence.
\par
In particular, the most natural divergence to use is the KL, for which we have $S_P(a) \leq \KLD{Q}{P} / \logtwo(a)$, applying which yields
\begin{align*}
\Exp[&\abs{I_n(\varphi) - I(\varphi)}] \\
&\quad\quad\quad\leq \norm{\varphi}_{L^2(Q)} \, \left(2^{-t/4} + 2 \sqrt{\frac{\KLD{Q}{P}}{\KLD{Q}{P} + t/2}}\right)\\
&\quad\quad\quad=\norm{\varphi}_{L^2(Q)} \, \left(2^{-t/4} + 2 \sqrt{\frac{L}{L + t/2}}\right)
\end{align*} 
\end{remark}

\begin{proof}[Proof of \Cref{thm:chatterjee-diaconis-strengthening}]
Write $a \defeq \exptwo(L+t/2)$ as in \cite{chatterjee2018sample}.
Now construct a `truncated' version of $\varphi$ via
\begin{align*}
h(x) = \varphi(x) \cdot \left(\frac{a}{r(x)} \wedge 1 \right) \, .
\end{align*}
We then follow the proof of \cite{chatterjee2018sample} mutatis mutandis.

First note that this gives
\begin{align*}
\varphi(x) - h(x) 
&= \varphi(x) \left( 1 - \Big(\frac{a}{r(x)} \wedge 1 \Big) \right) \\
&= \varphi(x) \, \Ind[r(x) > a] \,\frac{r(x) - a}{r(x)} 
\,\geq\, 0 \, .
\end{align*}
Then we can bound the approximation quality
\begin{align*}
\abs{I(\varphi) - I(h)} 
&= \Exp_{X \sim Q} \left[\varphi(X) - h(X)\right] \\
&= \int \varphi(x) \Ind[r(x) > a] \frac{r(x) - a}{r(x)} dQ(x) \\
&\leq \bigg[ \int \Ind[r(x) > a] \frac{r(x) - a}{r(x)} \varphi(x)^2 dQ(x) \\
&\qquad \cdot \int \Ind[r(x) > a] \frac{r(x) - a}{r(x)} dQ(x)\bigg]^{1/2} \\
&\leq \norm{\varphi}_{L^2(Q)} S_P(a)^{1/2}  \, .
\end{align*}

Now bound
\begin{align*}
\Exp \abs{I_n(\varphi) - I_n(h)}
&\leq \frac{1}{n} \sum_{i=1}^n \Exp(\abs{r(X_i) h(X_i) - r(X_i) \varphi(X_i)}) \\
&= \Exp_{X \sim P}[ r(X) \abs{\varphi(x) - h(x)}] \\
&= \Exp_{X \sim Q}[\varphi(X) - h(X) ]
\end{align*}
to see that this is also bounded by $\norm{\varphi}_{L^2(Q)} S_P(a)^{1/2}$.

Next we bound the variance
\begin{align*}
\Exp\left[\abs{I_n(h) - I(h)}\right]^2
&\leq \Exp\left[(I_n(h) - I(h))^2\right] \\
&= \frac{1}{n} \Var_{X \sim P}(r(X) h(X)) \\
&\leq \frac{1}{n} \Exp_{X \sim P}(r(X) h(X))^2 \\
&\leq \frac{1}{n} \Exp_{X \sim P}[(r(X)\wedge a) \varphi(X))^2] \\
&\leq \frac{a}{n} \Exp_{X \sim P}[r(X) \varphi(X)^2 ] \\
&= \frac{a}{n} \norm{\varphi}_{L^2(Q)}^2 \, .
\end{align*}
\par
Note $a/n = \exptwo(-t/2)$.
We can therefore conclude by the triangle inequality:
\begin{align*}
\Exp&(\abs{I_n(\varphi) - I(\varphi)}) \\
&\leq \Exp\left(\abs{I_n(\varphi) - I_n(h)} + \abs{I_n(h) - I(h)} + \abs{I(h) - I(\varphi)} \right) \\
&\leq \norm{\varphi}_{L^2(Q)} \, \left(2^{-t/4} + 2 S_P(a)^{1/2}\right) \,.
\end{align*}
\end{proof} 
\section{Lower bound on the entropy of the Block-Polyanskiy rejection sampler}
\label{app:lower_bound_on_rej_sampler_entropy}
Before we analyse the index, we introduce a small modification to the scheme to make it more efficient: if the rejection sampler doesn't terminate in the first $n$ steps, instead of picking an index at random, we pick the first index.
Thus, let $J$ be the index selected by this modified rule, i.e.\ 
\begin{align*}
J = \Ind[K \leq n]\cdot K + \Ind[K > n],
\end{align*}
where $K$ is modelling the runtime of the rejection sampler.
Note, that the modification doesn't change the output distribution $\widetilde{Q}_M$, but minimizes the entropy of the coded index.
\par
By the concavity of Shannon entropy, we have
\begin{align*}
\Ent[J] \geq \Prob[K \leq n] \cdot \Ent[K \mid K \leq n],
\end{align*}
\par
Let $B$ be the event that $K \leq n$.
Then, by the chain rule of Shannon entropy, we have
\begin{align*}
\Ent[K, B] &= \Ent[K \mid B] + \Ent[B] \\
&=\underbrace{\Ent[B \mid K]}_{=0} + \Ent[K].
\end{align*}
From the above, and by the definition of conditional entropy, we thus find
\begin{align*}
\Ent[K \mid &K \leq n] = \frac{\Ent[K] - \Ent[B] - \Prob[K > n] \Ent[K \mid K > n]}{\Prob[K \leq n]}
\end{align*}
Now, setting $\phi(x) = -x \logtwo x$, we get
\begin{align*}
\Ent[K \mid K > n] &= \sum_{k = n + 1}^\infty \phi(\Prob[K = k \mid K > n]) \\
&\stackrel{\text{(a)}}{=} \sum_{k = n + 1}^\infty \phi(\Prob[K = k - n]) \\
&=\sum_{k = 1}^\infty \phi(\Prob[K = k]) \\
&= \Ent[K]
\end{align*}
where equation (a) follows from the memoryless property of geometric random variables.
Plugging this identity back in the above equation, we get
\begin{align*}
\Ent[K \mid K \leq n] = \Ent[K] - \frac{\Ent[B]}{\Prob[K \leq n]},
\end{align*}
and thus
\begin{align*}
\Ent[J] &\geq \Prob[K \leq n] \cdot \Ent[K \mid K \leq n] \\
&= \Prob[K \leq n] \cdot \Ent[K] - \Ent[B]\\
&\stackrel{\text{(a)}}{\geq}
\left(1 - \frac{\epsilon}{2}\right)\logtwo M - 1 
\end{align*}
where inequality (a) follows from $\Ent[B] \leq 1$ since $B$ is a binary random variable, from $\Prob[K \leq n] \geq 1 - \epsilon/2$ by construction and that $\Ent[K] \geq -\logtwo p = \logtwo M$.
Finally, we note that when using the $\KLD{Q}{P}$ as our divergence, 
\begin{align*}
M = \frac{2}{1 - \epsilon} \ln\left(\frac{2}{\epsilon}\right)\exptwo\left(\frac{4 \KLD{Q}{P}}{\epsilon}\right) \, .
\end{align*}
Thus, for $\epsilon \leq 2/e$, we get
\begin{align*}
\Ent[J] \geq \left(1 - \frac{\epsilon}{2}\right) \frac{4 \KLD{Q}{P}}{\epsilon} - 1,
\end{align*}
and for small $\epsilon$, this will be far larger than $\KLD{Q}{P}$.

\end{document}